\documentclass[12pt,leqno,fleqn]{amsart}
\usepackage{latexsym}
\usepackage{amssymb}
\usepackage{amsxtra}
\usepackage{amsmath}
\usepackage{amsfonts}
\usepackage[dvips]{graphicx}
\usepackage{amsmath}
\usepackage{amssymb,amsmath,amsthm,amscd,mathrsfs,amsbsy,amsfonts}
\usepackage{pstricks,pst-node}
\usepackage{amsbsy,young,colortbl,cite}

\usepackage{graphicx}
\usepackage[metapost]{mfpic}
\usepackage{mflogo}

\usepackage{epstopdf}
\usepackage{graphics}

\def\@makecaption#1#2{\vskip\abovecaptionskip
  \sbox\@tempboxa{\small #1: #2}%
  \ifdim \wd\@tempboxa >\hsize \small #1: #2\par
  \else \global \@minipagefalse \hb@xt@\hsize{\hfil\box\@tempboxa\hfil}\fi
  \vskip\belowcaptionskip}
\newcommand{\cleqn}{\setcounter{equation}{0}}
\newcommand{\clth}{\setcounter{theorem}{0}}
\newcommand {\sectionnew}[1]{\section{#1}\cleqn\clth}
\newtheorem{theorem}{Theorem}[section]

\newtheorem{proposition}[theorem]{Proposition}

\newtheorem{definition}[theorem]{Definition}
\newtheorem{example}[theorem]{Example}

\theoremstyle{remark}

\renewcommand{\P}{\mathcal{P}}

\renewcommand{\L}{\mathcal{L}}

\newcommand{\Z}{\mathbb{Z}}

\renewcommand{\H}{\mathcal{H}}

\def\({\left(}
\def\){\right)}
\def\[{\begin{eqnarray*}}
\def\]{\end{eqnarray*}}
\def\d{\partial}
\def\ep{\epsilon}
\def\La{\Lambda}
\def\la{\lambda}

\def\om{\omega}

\begin{document}

\title{ \Large \bf \boldmath\ \\ Regular solution and lattice miura transformation  of  bigraded Toda Hierarchy*}
\author{\large  Chuanzhong LI\dag\ddag
, Jingsong HE\dag }

\thanks{ \begin{tabular}{@{}r@{}p{13.4cm}@{}}
& Manuscript received  \\ 
$\dag$ & Department of Mathematics,  Ningbo University, Ningbo, 315211, China\\
&{E-mail:lichuanzhong@nbu.edu.cn} \\
$^{\ddag}$ & Department of Mathematics,  USTC, Hefei, 230026, China\\
&{Corresponding email:hejingsong@nbu.edu.cn}\\
 $^{\ast}$ & Project supported by
 Natural Science Foundation  of Zhejiang Province under Grant No.LY12A01007, the National Natural Science Foundation of China under Grant  No.11201251, 10971109, K. C. Wong Magna Fund in
Ningbo University.
\end{tabular}}
\texttt{}

\date{}
\begin{abstract}
In this paper, we give finite dimensional exponential  solutions of the bigraded Toda Hierarchy(BTH). As
 an specific example of exponential solutions of the BTH, we consider a regular solution for the $(1,2)$-BTH with $3\times 3$-sized Lax matrix, and
 discuss some geometric structure of the solution from which the difference between $(1,2)$-BTH and original Toda hierarchy is shown. After this, we construct another kind of Lax representation of $(N,1)$-bigraded Toda hierarchy($(N,1)$-BTH) which does not use the fractional operator of Lax operator. Then we introduce lattice Miura transformation of $(N,1)$-BTH which leads to equations depending on one field, meanwhile we give some specific examples which contains Volterra lattice equation(an useful ecological competition model).
\end{abstract}
\maketitle
\noindent
Keywords:  Regular solution, lattice miura transformation, bigraded Toda Hierarchy, moment polytope, Volterra lattice.\\
2000 Mathematics Subject Classifications   37K05, 37K10, 37K20

\allowdisplaybreaks
 \setcounter{section}{0}

 \sectionnew{Introduction}
   The Toda lattice hierarchy introduced by M. Toda \cite{Toda,Todabook}
  is a completely integrable system and has important applications in many different fields such
 as classical and quantum fields  theory.   It is well-known that
 the Toda lattice equation \cite{UT} can be reduced from the two-dimensional
Toda hierarchy. Adding additional logarithm flows to the Toda lattice hierarchy,
it becomes the extended Toda hierarchy\cite{CDZ} which governs the Gromov-Witten invariant of $CP^1$.
The bigraded Toda hierarchy (BTH) of $(N,M)$-type(or simply the $(N,M)$-BTH) is the generalized Toda lattice hierarchy
whose infinite Lax matrix has $N$ upper and $M$ lower nonzero diagonals.
The BTH can be seen as a natural extension of the original Toda lattice hierarchy which is just of $(1,1)$-type. The BTH can be also treated as a general reduction of the two-dimensional Toda lattice hierarchy. The extended bigraded Toda
hierarchy(EBTH) is the extension of the BTH which includes additional logarithm flows\cite{C}. The dispersionless version of extended bigraded
Toda hierarchy was firstly introduced by S. Aoyama, Y. Kodama in
\cite{KodamaCMP}.  Later dispersive extended bigraded Toda hierarchy was
introduced by Gudio Carlet \cite{C} who hoped that the EBTH might
also be relevant to the theory of Gromov-Witten invariants.

 We generalized the
Sato theory to the EBTH  and give the HBEs of EBTH in \cite{ourJMP}.
The close relation of the BTH and the two-dimensional Toda hierarchy becomes a great motivation for us to consider the solutional
structure of the BTH. In paper \cite{solutionBTH},
 we prove the BTH( bigraded toda hierarchy) has an  equivalent relation between $(N,M)$-BTH
 (whose infinite Lax matrix has $N$ upper and $M$ lower nonzero diagonals)
and $(M,N)$-BTH. In paper \cite{ourBlock,dispBTH},  the BTH is proved to have a natural Block type Lie algebraic symmetry and so is dispersionless BTH.
As we know, $(N,M)$-BTH is equivalent to $(N,M)$-band bi-infinite matrix-formed Toda hierarchy, so we consider its reduction,
 i.e. the semi-finite and finite matrix form of the BTH.  Then
we give solutions of the BTH  using orthogonal polynomials in matrix
form. Some rational solutions of the BTH and corresponding
Young diagrams were also given in \cite{solutionBTH}. But there is one missing part is about the regular exponential solutions of BTH. Therefore from the general structure of
solutions of  the BTH, regular  exponential solution which only
depend on primary time variables  will be introduced in this article. Also we will tell the difference between $(1,2)$-BTH and original Toda hierarchy from the orbit of flows in the graph of diagonal elements. Because its structure of Flag manifold  similar as original Toda hierarchy, a geometric description using moment polytope \cite{topology} whose vertices correspond to solutions of the BTH.

Comparing the equations for primary flows of the $(N,1)$-BTH with equations constructed in \cite{svinin}, we find that they have very close relation which will be shown in detail in the following sections.

The paper is organized as follows. In Section 2, the definition of the BTH and its tau function are given. In Section 3,
 the exponential  solutions of the BTH will be given where we also  consider the finite dimensional exponential solutions and further survey into the  $(1,M)$-BTH. To see the geometry of the $(N,M)$-BTH, we consider
  the regular solution for the $(1,2)$-BTH with $3\times 3$ Lax matrix, and
 discuss the geometric structure of the solution, i.e. the moment polytope for the $(1,2)$-BTH in Section 4.  In Section 5,  some primary flows of $(N,1)$-BTH  will be introduced. In Section 6,  we construct another kind of Lax representation of
the BTH. In Section 7,  lattice miura transformation of the BTH will be given, meanwhile, some concrete examples will be shown in detail.
After that conclusion and discussion are devoted to the last section.

\section{The bigraded Toda hierarchy (BTH)}\label{sec:BTH}
Firstly, the interpolated bigraded Toda hierarchy will be introduced firstly.
The Lax operator of the BTH is given by the Laurent polynomial of shift matrix $\Lambda$ \cite{C}
\begin{equation}\label{LBTH}
\L:=\Lambda^{N}+u_{N-1}\Lambda^{N-1}+\cdots+u_0+\cdots + u_{-M}
\Lambda^{-M},
\end{equation}
where
   $N,M \geq1$, $\Lambda$ represents the shift operator with $\Lambda:=e^{\epsilon\partial_x}$ and $``\epsilon"$ is
called the string coupling constant, i.e. for any function $ f(x)$
\begin{equation*}
\Lambda f(x)=f(x+\epsilon).
\end{equation*}

 The $\L$ can be
written in two different ways by dressing the shift operator
\begin{equation}\label{two dressing}
\L=\P_L\Lambda^N\P_L^{-1} = \P_R \Lambda^{-M}\P_R^{-1},
\end{equation}
where the dressing operators have the form,
 \begin{align}
 \P_L&=1+w_1\Lambda^{-1}+w_2\Lambda^{-2}+\ldots,
\label{dressP}\\[0.5ex]
 \P_R&=\tilde{w_0}+\tilde{w_1}\Lambda+\tilde{w_2}\Lambda^2+ \ldots.
\label{dressQ}
\end{align}
Note that $\P_L$ is lower triangular Matrix, $\P_R$  is upper triangular Matrix.

Eq.\eqref{two dressing} are quite important because it gives the reduction condition from the two-dimensional Toda lattice hierarchy.
 The
pair is unique up to multiplying $\P_L$ and $\P_R$ from the right
 by operators in the form  $1+
a_1\Lambda^{-1}+a_2\Lambda^{-2}+...$ and $\tilde{a}_0 +
\tilde{a}_1\Lambda +\tilde{a}_2\Lambda^2+\ldots$ respectively with
coefficients independent of $x$. Given any difference operator $A=
\sum_k A_k \Lambda^k$, the positive and negative projections are
defined by $A_+ = \sum_{k\geq0} A_k \Lambda^k$ and $A_- = \sum_{k<0}
A_k \Lambda^k$.

To write out  explicitly the Lax equations  of the BTH,
  fractional powers $\L^{\frac1N}$ and
$\L^{\frac1M}$ are defined by
\begin{equation}
  \notag
  \L^{\frac1N} = \Lambda+ \sum_{k\leq 0} a_k \Lambda^k , \qquad \L^{\frac1M} = \sum_{k \geq -1} b_k
  \Lambda^k,
\end{equation}
with the relations
\begin{equation}
  \notag
  (\L^{\frac1N} )^N = (\L^{\frac1M} )^M = \L.
\end{equation}
$\L^{\frac{1}{N}}$  and $\L^{\frac{1}{M}}$ are lower Heisenberg
triangular Matrix and upper triangular Matrix respectively.

Acting on free function, these two fraction powers can be seen as two different locally expansions around zero and infinity respectively.
It was  stressed that $\L^{\frac1N}$ and $\L^{\frac1M}$ are two
different operators even if $N=M(N, M\geq 2)$ in \cite{C} due to two different dressing operators. They can also be
expressed as following
\begin{equation}
\notag
  \L^{\frac1N} = \P _{L}\Lambda\P_{L}^{-1}, \qquad \L^{\frac1M} = \P_{R}\Lambda^{-1} \P_{R}^{ -1}.
\end{equation}
Let  us now define the following operators for the generators of the BTH flows,
\begin{equation}
  B_{\gamma , n} :=\left\{
\begin{array}{llll}
 \displaystyle{ \L^{n+1-\frac{\alpha-1}{N}} }&\qquad {\rm if}\quad \gamma=\alpha=1,2,\ldots,N,\\[2.0ex]
 \displaystyle{   \L^{n+1+\frac{\beta}{M}}} &\qquad{\rm if}\quad\gamma=\beta = -M+1,\ldots,-1,0,
\end{array}\right.
\end{equation}
\begin{definition}
  Bigraded Toda hierarchy(BTH) in the Lax representation is given by the set of
infinite number of flows defined by(\cite{solutionBTH})
\begin{equation}\label{BTHLax}
\frac{\partial \L}{\partial t_{\gamma, n}} =\left\{
\begin{array}{lll}
[ (B_{\alpha,n})_+, \L ], &\quad {\rm if} \quad\gamma=\alpha= 1,2,\ldots,N,\\[1.5ex]
[ -(B_{\beta,n})_-, \L ], &\quad{\rm if}\quad \gamma=\beta=   -M+1,\ldots,-1,0.
\end{array}\right.
\end{equation}
\end{definition}
Sometimes we denote $\frac{\partial }{\partial t_{\gamma, n}}$ as $\partial_{\gamma, n}$ in this paper and denote $\{t_{\gamma, n},\gamma=\alpha= 1,2,\ldots,N,\},\{t_{\gamma, n},\gamma=\beta=   -M+1,\ldots,-1,0,\}$ as $t_{\alpha},t_{\beta}$ respectively.
We call the flows for $n=0$ the {\it primaries} of the BTH and the time variables $t_{\gamma, n}$ ($n=0$) are primary time variables.
The original tridiagonal Toda  hierarchy corresponds to the case with $N=M=1$.

\subsection{Tau function and band structure}
According to paper \cite{ourJMP},
a function $\tau$  depending only on the dynamical variables $t$ and
$\epsilon$ is called the  {\em  tau-function of BTH} if it
provides symbols related to wave operators as following,
\begin{eqnarray}\label{pltau}P_L: &=&
1+\frac{w_1}{\lambda}+\frac{w_2}{\lambda^2}+\ldots : =\frac{ \tau
(x, t-[\lambda^{-1}]^N;\epsilon) }
     {\tau (x,t;\epsilon)},\\\label{pl-1tau}
P_L^{-1}:& = &1+\frac{w_1'}{\lambda}+\frac{w_2'}{\lambda^2}+\ldots
: = \frac{\tau (x+\epsilon,
t+[\lambda^{-1}]^N;\epsilon) }
     {\tau (x+\epsilon,t;\epsilon)},\\\label{prtau}
P_R:&= &\tilde w_0 + \tilde w_1\lambda+ \tilde w_2\lambda^{2}
+\ldots : = \frac{ \tau
(x+\epsilon,t+[\lambda]^M;\epsilon)}
     {\tau(x,t;\epsilon)},\\\label{pr-1tau}
     P_R^{-1}:&= &\tilde w_0 '+ \tilde w_1'\lambda+ \tilde
w_2'\lambda^{2} +\ldots : = \frac{
     \tau (x,t-[\lambda]^M; \epsilon)}
     {\tau(x+\epsilon,t;\epsilon)},
     \end{eqnarray}
     where $[\lambda^{-1}]^N$ and $[\lambda]^{M}$  are defined by
\begin{equation} \notag
 [\lambda^{-1}]^{N}_{\gamma,n} :=
\begin{cases}
  \frac{\lambda^{-N(n+1-\frac{\gamma-1}{N})}}{N(n+1-\frac{\gamma-1}{N})}, &\gamma=N,N-1,\dots 1,\\
 0, &\gamma = 0, -1\dots -(M-1),
  \end{cases}
\end{equation}
\begin{equation} \notag
[\lambda]^{M}_{\gamma,n} :=
\begin{cases}
0, &\gamma=N, N-1,\dots 1,\\
\frac{\lambda^{M(n+1+\frac{\beta}{M})}}{M(n+1+\frac{\beta}{M})}
, &\gamma = 0, -1, \dots -(M-1).
  \end{cases}
\end{equation}

Then we get
\begin{eqnarray}\label{pltau}
&P_L: =\displaystyle{\sum_{n=0}^{\infty}\frac{ P_n(-\hat\d_L)\tau
(x,t;\epsilon) }
     {\tau (x,t;\epsilon)}\la^{-n}},\qquad
&P_L^{-1}:= \sum_{n=0}^{\infty}\frac{P_n(\hat\d_L)\tau (x+\epsilon,
t;\epsilon) }
     {\tau (x+\epsilon,t;\epsilon)}\la^{-n},\\\label{prtau}
&P_R:= \displaystyle{\sum_{n=0}^{\infty}\frac{ P_n(\hat\d_R)\tau
(x+\epsilon,t;\epsilon)}
     {\tau(x,t;\epsilon)}\la^{n}}, \qquad
   &  P_R^{-1}:= \sum_{n=0}^{\infty}\frac{P_n(-\hat\d_R)
     \tau (x,t; \epsilon)}
     {\tau(x+\epsilon,t;\epsilon)}\la^{n},
     \end{eqnarray}
where Schur polynomial $P_{k}({\hat \d})$ is defined by
\begin{equation}\label{shurfunction}
e^{\sum^{\infty}_{k=1}\frac{1}{k}\d_kz^k}=\sum^{\infty}_{k=0}P_{k}({\hat
\d})z^{k}, \quad {\hat \d}=(\d_1, \frac{1}{2}\d_2, \frac{1}{3}\d_3,
\frac{1}{4}\d_4, ... ).
\end{equation}
Here the operators $\hat\partial_L$ and $\hat\partial_R$ are defined by
\begin{align*}
\hat \d_L&=\left\{\frac{1}{N(n+1-\frac{\alpha-1}{N})}\d_{t_{\alpha,n}}:  1\leq\alpha\leq N\right\}\\[1.0ex]
\hat \d_R&=\left\{\frac{1}{M(n+1+\frac{\beta}{M})}\d_{t_{\beta,n}}: -M+1\leq\beta\leq0\right\}.
\end{align*}

The dressing operators $\P_L$ and $\P_R$ can be expressed by function $\tau(x,t;\epsilon)$:
 \begin{align}\label{pltau}\P_L&=\sum_{n=0}^{\infty}\frac{ P_n(-\hat\d_L)\tau
(x,t;\epsilon) }
     {\tau (x,t;\epsilon)}\La^{-n},\qquad
\P_L^{-1} =  \sum_{n=0}^{\infty}\La^{-n}\frac{P_n(\hat\d_L)\tau (x+\epsilon,
t;\epsilon) }
     {\tau (x+\epsilon,t;\epsilon)},\\\label{prtau}
\P_R&=  \sum_{n=0}^{\infty}\frac{ P_n(\hat\d_R)\tau
(x+\epsilon,t;\epsilon)}
     {\tau(x,t;\epsilon)}\La^{n},\qquad
     \P_R^{-1}
     =  \sum_{n=0}^{\infty}\La^{n}\frac{P_n(-\hat\d_R)
     \tau (x,t; \epsilon)}
     {\tau(x+\epsilon,t;\epsilon)}.
     \end{align}

 One can then find the explicit form of the coefficients $u_i(x,t)$ of the operator $\mathcal{L}$ in terms of the $\tau$-function using eq.\eqref{two dressing} as \cite{DiscreteKP, adler92},
    \begin{align}\label{ui with tau}
 u_i(x,t)&=\frac{P_{N-i}(\hat D_L)\tau (x+(i+1)\ep,
t;\epsilon)\circ \tau
(x,t;\epsilon) }
     {\tau (x,t;\epsilon)\,\tau (x+(i+1)\ep,t;\epsilon)}=\frac{ P_{M+i}(\hat D_R)\tau
(x+\epsilon,t;\epsilon)\circ\tau (x+i\ep,t; \epsilon)}
     {\tau(x,t;\epsilon)\,\tau(x+(i+1)\ep,t;\epsilon)},
 \end{align}
 where $\hat D_L$ and $\hat D_R$ are just the Hirota derivatives corresponding to $\hat \d_L$ and $\hat \d_R$ respectively.

As we all know, interpolated BTH is equivalent to bi-infinite or semi-infinite matrix-formed BTH.
Because what we will consider next is the matrix-formed bigraded Toda hierarchy, following equivalent definitions in matrix form are introduced,
$$\Lambda:=(E_{i,i+1})_{i\in\mathbb{Z_+}}, \ \ u_i:=diag(u_{i,1},u_{i,2},u_{i,3},\dots).$$
After the following transformation
 $u_i(x):=u_{i,j}:=a_{j,j+i}$, the matrix representation of $\L$ can be expressed by $(a_{i,j})_{i,j\geq 1}$ with
\begin{equation}\label{aij}
a_{i,j}(t)=\frac{P_{i-j+N}(\hat D_L)\tau_j\circ\tau_{i-1}}{\tau_{i-1}\tau_j}
=\frac{P_{j-i+M}(\hat D_R)\tau_i\circ\tau_{j-1}}{\tau_{i-1}\tau_j}.
\end{equation}

This immediately imply
\begin{align*}
a_{i,j}&=0,\qquad {\rm if}\quad j<-M+i\quad{\rm or}\quad j>N+i.
\end{align*}
That shows that the Lax matrix $\L$ has the $(N,M)$-band structure.

\sectionnew{Exponential solutions of the BTH}
In paper \cite{solutionBTH}, the rational solutions of the BTH were
introduced already.  In this section, we will introduce the exact (regular)
solutions of the BTH, i.e. the non-negative exponential solutions.

The tau functions of the two-dimensional Toda lattice hierarchy \cite{adler92} can be expressed by
\begin{equation}\label{tau n2Dtoda}
\tau_{i}=\left|\begin{array}{ccccc}\bar C_{0,0} & \bar C_{0,1}&\dots &\bar C_{0,i-1}\\ \bar C_{1,0} &\bar C_{1,1} &\dots &\bar C_{1,i-1}\\ \dots &\dots &\dots &\dots \\
\bar C_{i-1,0} &\bar C_{i-1,1}&\dots&\bar C_{i-1,i-1}\end{array}\right|,
\end{equation}
where $\bar C_{i,j}$ can have be in form of following inner product using arbitrary density function $\rho(\la,\mu)$
\begin{eqnarray*}
\bar C_{i,j}&=&\langle\la^{i},\mu^{j}\rho(\la,\mu) \rangle.
\end{eqnarray*}
Here the inner product can be chosen as following integral representation
\begin{eqnarray*}
\bar C_{i,j}&=&\int\int\rho(\la,\mu)\la^{i}\mu^{j}e^{\sum_{n=0}^{\infty}x_n\la^n+\sum_{n=0}^{\infty}y_n\mu^n}d \la d\mu\\
&=&\sum_{k,l=0}^{\infty}\bar c_{i,j,k,l}P_k(x)P_l(y),
\end{eqnarray*}
where $P_k(x)$ and $P_l(y)$ are Schur function introduced in eq.\eqref{shurfunction}.

We should note here that the coefficients $\bar c_{i,j,k,l}$ are totally independent.

As the original tridiagonal Toda lattice is $(1,1)$-reduction of the two-dimensional Toda lattice hierarchy. Therefore to get the solution of the  tridiagonal Toda lattice hierarchy, we need to add factor
$\delta(\la-\mu)$ under the integral in the definition of $\bar C_{i,j}$, i.e. the element $\bar C_{i,j}$ in tau function of the  tridiagonal Toda lattice hierarchy \cite{ISo} becomes
\begin{eqnarray}\label{11BTH}
\int\int\rho(\la,\mu)\delta(\la-\mu)\la^{i}\mu^{j}e^{\sum_{n=0}^{\infty}x_n\la^n+\sum_{n=0}^{\infty}y_n\mu^n}d \la d\mu,
\end{eqnarray}
which can further lead to
\begin{eqnarray}\label{11BTH2}
\int\rho(\la,\la)\la^{i+j}e^{\sum_{n=0}^{\infty}(x_n+y_n)\la^n}d \la.
\end{eqnarray}
After changing time variables $x,y$  to $t_{\alpha},t_{\beta}$ in BTH, eq.\eqref{11BTH2} become a new function
$$\int\rho(\la,\la)\la^{i+j}e^{\xi_L(\la,t_{\alpha})+\xi_R(\la,t_{\beta})}d \la$$ which corresponds to  $(1,1)$-BTH.

Denote $\omega_N$ and $\omega_M$ as the N-th root and M-th root of unit.
For $(N,M)$-BTH which is a
 generalization of the tridiagonal Toda lattice hierarchy, new function $C_{i,j}$(new form of $\bar C_{i,j}$) has the following form

\begin{eqnarray*}
C_{i,j}&=&\int\int\rho(\la,\mu)\delta(\la^N-\mu^M)\la^{i}\mu^{j}e^{\xi_L(\la,t_{\alpha})+\xi_R(\mu,t_{\beta})}d \la d\mu\\
&=&\sum_{p=0}^{N-1}\sum_{q=0}^{M-1}\int\rho(\omega_N^p\la^{\frac 1N},\omega_M^q\la^{\frac 1M})(\omega_N^p\la^{\frac 1N})^{i}
(\omega_M^q\la^{\frac 1M})^{j}e^{\xi_L(\omega_N^p\la^{\frac 1N},t_{\alpha})
+\xi_R(\omega_M^q\la^{\frac 1M},t_{\beta})}d \la.
\end{eqnarray*}

Therefore tau functions of the BTH  can be explicitly written in the form \cite{solutionBTH}
\begin{equation}\label{tau n}
\tau_{i}=\left|\begin{array}{ccccc}C_{0,0} & C_{0,1}&\dots &C_{0,i-1}\\ C_{1,0} &C_{1,1} &\dots &C_{1,i-1}\\ \dots &\dots &\dots &\dots \\
C_{i-1,0} &C_{i-1,1}&\dots&C_{i-1,i-1}\end{array}\right|.
\end{equation}
If we consider the case in  finite dimension(dimension is n), i.e. \begin{eqnarray}\rho(\omega_N^p\la^{\frac 1N},\omega_M^q\la^{\frac 1M})=
\sum_{k=1}^{n}\rho_0(\omega_N^p\la^{\frac 1N},\omega_M^q\la^{\frac 1M})\delta(\la-\la_k),\end{eqnarray} then
\begin{eqnarray}
C_{i,j}&=&\sum_{p=0}^{N-1}\sum_{q=0}^{M-1}\sum_{k=1}^{n}\rho_0(\omega_N^p\la_k^{\frac 1N},\omega_M^q\la_k^{\frac 1M})
(\omega_N^p\la_k^{\frac 1N})^{i}(\omega_M^q\la_k^{\frac 1M})^{j}e^{\xi_L(\omega_N^p\la^{\frac 1N},t_{\alpha})
+\xi_R(\omega_M^q\la^{\frac 1M},t_{\beta})} .
\end{eqnarray}
After denoting $C_{0,0}$ as $\tau_1$, we can rewrite tau functions
into the following bi-directional Wronskian form
\begin{equation}\label{taunderivative}
\tau_{i}=\left|\begin{array}{ccccc}\tau_1 & \d_{t_{-M+1,0}}\tau_1&\dots &\d^{i-1}_{t_{-M+1,0}}\tau_1\\ \d_{t_{N,0}}\tau_1 &\d_{t_{N,0}}
\d_{t_{-M+1,0}}\tau_1 &\dots &\d_{t_{N,0}}\d^{i-1}_{t_{-M+1,0}}\tau_1\\ \dots &\dots &\dots &\dots \\
\d^{i-1}_{t_{N,0}}\tau_1 &\d^{i-1}_{t_{N,0}}
\d_{t_{-M+1,0}}\tau_1&\dots&\d^{i-1}_{t_{N,0}}\d^{i-1}_{t_{-M+1,0}}\tau_1\end{array}\right|.
\end{equation}
Then $\tau_1$ has the following form
 \begin{eqnarray}
\tau_1&=&\sum_{p=0}^{N-1}\sum_{q=0}^{M-1}\sum_{k=1}^{n}
\rho_0(\omega_N^p\la_k^{\frac 1N},\omega_M^q\la_k^{\frac 1M})e^{\xi_L(\omega_N^p\la^{\frac 1N}, t_{\alpha})
+\xi_R(\omega_M^q\la^{\frac 1M}, t_{\beta})}.
\end{eqnarray}
For $(N,M)$ case,  for each $m\leq n$
\begin{eqnarray*}\notag
\tau_m=
&&\sum_{1\leq i_1\leq i_2\leq\dots \leq i_m\leq
n}\sum_{i_j',i_k'=1}^N \sum_{\bar i_j,\bar i_k=1}^M\prod_{1\leq
k<j\leq m}
(\omega_N^{i_j'}\omega_{i_j}^{\frac1N}-\omega_N^{i_k'}\om_{i_{k}}^{\frac1N})
\vec{}(\omega_M^{\bar i_j}\omega_{i_j}^{\frac1M}-\omega_M^{\bar
i_k}\om_{i_{k}}^{\frac1M})\\
&&\prod _{j,k,l=1}^m
\rho_0(\omega_N^{i_k'}\la_{i_j}^{\frac 1N},\omega_M^{\bar
i_l}\la_{i_j}^{\frac 1M})E_{i_j,i'_k,\bar i_l},
\end{eqnarray*}
where
\begin{eqnarray}
E_{i_j,i'_k,\bar i_l}=e^{\xi_L(\omega_N^{i'_k}\la_{i_j}^{\frac 1N}, t_{\alpha})
+\xi_R(\omega_M^{\bar i_l}\la_{i_j}^{\frac 1M}, t_{\beta})}.
\end{eqnarray}
We can find it  there are  ${n \choose m}(NM)^m $ terms in  $\tau_m$.

The theory above is about $(N,M)$-BTH. To see it clearly, we will
further consider one kind of specific example, i.e. the $(1,M)$-BTH clearly in the following.

Before that firstly we will consider the finite-sized
Lax matrix of the $(1,M)$-BTH.

If Lax matrix is supposed to have n different eigenvalues, i.e. there exists matrix
$\Phi$ s.t.
\begin{eqnarray}\notag
\Phi L\Phi^{-1}=dig(\la_1, \la_2,\dots, \la_n).
\end{eqnarray}
Then from the theory on the $(N,M)$-BTH, we can get the first tau
function of the $(1,M)$-BTH as following
\begin{eqnarray}
\tau_1&=&\sum_{q=0}^{M-1}\sum_{k=1}^{n}\rho_0(\la_k,\omega_M^q\la_k^{\frac 1M})
e^{\xi_R(\omega_M^q\la_k^{\frac 1M}, t_{\beta})}.
\end{eqnarray}
Similarly for every integer $m\leq n$, tau function $\tau_m$ has following form
\begin{eqnarray*}
&&\tau_m=\\
&&
\sum_{1\leq i_1\leq i_2\leq\dots \leq i_m\leq n}\sum_{i_j',i_k'=1}^M\prod_{1\leq k<j\leq m}
(\omega_{i_j}-\om_{i_k})\prod_{1\leq k<j\leq m}
(\omega_M^{i_j'}\omega_{i_j}^{\frac1M}-\omega_M^{i_k'}
\om_{i_{k}}^{\frac1M})
\prod _{j=1}^m\rho_0(\la_{i_j},\omega_M^{i_j'}\la_{i_j}^{\frac 1M})E_{i_j,i'_{j}},
\end{eqnarray*}
where
\begin{eqnarray*}E_{i_j,i'_{j}}=e^{\xi_R(\omega_M^{i'_{j}}\la_{i_j}^{\frac 1M}, t_{\beta})}.
\end{eqnarray*}
Considering the primary dependence, $\tau_1$ can be written as
\begin{eqnarray}
\tau_1=\sum_{q=0}^{M-1}\sum_{k=1}^{n}\rho_0(\la_k,\omega_M^q\la_k^{\frac 1M})e^{\sum_{s=1}^M(\omega_M^q\la_k^{\frac 1M})^{ s}t_{s-M,0}} .
\end{eqnarray}

To see it more clearly, we will further consider a specific example in
the following section, i.e. the exponential solution of the $(1,2)$-BTH.

\sectionnew{Regular solution of $(1,2)$-BTH}
One can obtain the general finite dimensional solution for the $(N,M)$-BTH with exponential
functions.  However most of the solutions are complex and have singular points. In this section, we will consider the exponential
 solution particularly the regular solutions of the
  $(1,2)$-BTH. By this regular solution, we see the difference between $(1,2)$-BTH and original Toda hierarchy from a geometric viewpoint.

For $(1,2)$ case, the solution can be expressed as
\\
\begin{eqnarray}
\tau_1&=&\sum_{k=1}^{n}\rho_0(\la_k,\la_k^{\frac 12})
e^{\xi_R(\la_k^{\frac 12},t_{\beta})}+\rho_0(\la_k,\omega_2\la_k^{\frac 12})
e^{\xi_R(\omega_M\la_k^{\frac 12},t_{\beta})},
\end{eqnarray}
where $\omega_2=-1.$

If we only consider the primary dependence which means we let tau function only depend on primary time variables, then the first solution $\tau_1$ in form of one by one matrix has following form
\begin{eqnarray}
\tau_1&=&\sum_{k=1}^{n}\rho_0(\la_k,\la_k^{\frac 12})
e^{\omega_{k}^{\frac12}\la_k^{\frac 12}t_{-1,0}+\omega_{k}\la_kt_{0,0}}+\rho_0(\la_k,-\la_k^{\frac 12})
e^{-\omega_{k}^{\frac12}\la_k^{\frac 12}t_{-1,0}+\omega_{k}\la_kt_{0,0}}.
\end{eqnarray}

Let us assume that the Lax matrix is semi-simple and has  distinct eigenvalues,
$(\lambda_1,\lambda_2,\lambda_3)$.
 If the $\lambda_{k},k=1,2,3$ is negative,
the  real solution  can be written as the following form
\begin{eqnarray}\tau_1&=&
\sum_{k=1}^n\rho_0(\la_k,\la_k^{\frac 12})e^{i|\lambda_{k}|^{\frac12}t_{-1,0}+\lambda_{k}t_{0,0}}
+\rho_0(\la_k,-\la_k^{\frac 12})e^{-i|\lambda_{k}|^{\frac12}t_{-1,0}+\lambda_{k}t_{0,0}}\\
&=&
\sum_{k=1}^n A'_{k}
 \cos(|\lambda_{k}|^{\frac12}t_{-1,0}+\theta_k)e^{\lambda_{k}(t_{0,0}+t_{1,0})},
 \end{eqnarray}
 where
 $A'_{k}:=\sqrt{\rho_0(\la_k,\la_k^{\frac 12})^2+\rho_0(\la_k,-\la_k^{\frac 12})^2}$
 and $\theta_k$ depend on $\{\rho_0(\la_k,\la_k^{\frac 12}),\rho_0(\la_k,-\la_k^{\frac 12}),|\lambda_{k}|^{\frac12}t_{-1,0}\}.$
  This is in fact a periodic solution about $t_{-1,0}$ time variable and it has singular points.\\
 If we set $0<\lambda_{1}<\lambda_{2}<\lambda_{3}, \rho_0\geq0$, it will lead to regular solutions.
In this case as a simple but an interesting example,
 we will consider a regular solution for the $(1,2)$-BTH with $3\times 3$-sized  Lax matrix, and
 discuss some geometric structure of the solution in the following part.
Then the regular function $\tau_1$ is given by
\begin{equation*}
\tau_1=
\sum_{k=1}^3 A_{k}
\cosh(\lambda_{k}^{\frac12}t_{-1,0}+\theta_k)e^{\lambda_{k}(t_{0,0}+t_{1,0})}=\sum_{k=1}^3 C_{k}
 E_{k},
\end{equation*}
 where $A_k$ and $\theta_k$ are arbitrary constants, $C_{k}:=A_{k}
 \cosh(\lambda_{k}^{\frac12}t_{-1,0}+\theta_k)$ and $E_k:=e^{\lambda_k(t_{0,0}+t_{1,0})}$.
 We also write  $S_{k}:=A_{k} \sinh(\lambda_{k}^{\frac12}t_{-1,0}+\theta_k)$ in the following part.

For $\tau_1$ being positive definite, $A_i>0$ is supposed to hold. Also we set all $\theta_k=0$ (this is
necessary for $\tau_k$ being sign definite).
Then the second tau function$\tau_2$ and the third tau function $\tau_3$ in eq.\eqref{taunderivative} are given by
\begin{align*}
\tau_2&=
\left|\begin{matrix}
\tau_1&\d_{1,0}\tau_1\\
\d_{-1,0}\tau_1&\d_{-1,0}\d_{1,0}\tau_1
\end{matrix}\right|=
\left|\left(\begin{matrix}
C_1E_{1}&C_2E_{2}&C_3E_{3}\\
\lambda_1^{\frac12}S_1E_{1}&\lambda_2^{\frac12}S_2E_{2}&\lambda_3^{\frac12}S_3E_{3}\\
\end{matrix}\right)
\left(\begin{matrix}
1&\lambda_1\\
1&\lambda_2\\
1&\lambda_3
\end{matrix}\right)\right|\\
&=\sum_{i,j=1,i< j}^3 (C_{i}S_{j}\lambda_{j}^{\frac12}-C_{j}S_{i}\lambda_{i}^{\frac12})(\lambda_j-\lambda_i)E_{i}E_{j},
\end{align*}

\begin{align*}
\tau_3&=\left|\begin{matrix}
\tau_1&\d_{1,0}\tau_1&\d_{1,0}^2\tau_1\\
\d_{-1,0}\tau_1&\d_{-1,0}\d_{1,0}\tau_1&\d_{-1,0}\d_{1,0}^2\tau_1\\
\d_{-1,0}^2\tau_1&\d_{-1,0}^2\d_{1,0}\tau_1&\d_{-1,0}^2\d_{1,0}^2\tau_1
\end{matrix}\right|\\
&=
\left|\left(\begin{matrix}
C_1E_{1}&C_2E_{2}&C_3E_{3}\\
\lambda_1^{\frac12}S_1E_{1}&\lambda_2^{\frac12}S_2E_{2}&\lambda_3^{\frac12}S_3E_{3}\\
\lambda_1C_1E_{1}&\lambda_2C_2E_{2}&\lambda_3C_3E_{3}\\
\end{matrix}\right)
\left(\begin{matrix}
1&\lambda_1&\lambda_1^2\\
1&\lambda_2&\lambda_2^2\\
1&\lambda_3&\lambda_3^2
\end{matrix}\right)\right| \\
&=\left(\sum_{i\to j\to k}\lambda_i^{\frac12}|\lambda_j-\lambda_k|S_iC_jC_k\right)\prod_{i<j}(\lambda_j-\lambda_i)E_1E_2E_3,
\end{align*}
where $\sum_{i\to j\to k}$ implies the sum over the cyclic permutation on $\{1,2,3\}$.
Therefore $\tau_1$ is always positive.

For  $\tau_2$, taking derivative can imply
\begin{align}\label{derivativepositive}
&\d_{t_{-1,0}}(C_{i}S_{j}\lambda_{j}^{\frac12}-C_{j}S_{i}\lambda_{i}^{\frac12})\\\nonumber
&=S_{i}S_{j}(\lambda_{i}\lambda_{j})^{\frac12}+C_{i}C_{j}\lambda_{j}-S_{j}S_{i}(\lambda_{i}\lambda_{j})^{\frac12}-C_{j}C_{i}\lambda_{i}\\\nonumber
&=C_{i}C_{j}(\lambda_{j}-\lambda_{i})>0,
\end{align}
which means $\tau_2$ always increase with variable $t_{-1,0}$.
Because when $t_{-1,0}=0$,
\begin{align*}
C_{i}S_{j}\lambda_{j}^{\frac12}-C_{j}S_{i}\lambda_{i}^{\frac12}=0,
\end{align*}
 when $t_{-1,0}>0$, $\tau_2$ is always positive.
Because
\begin{align*}
\lim _{t_{-1,0}\rightarrow 0}\frac{\d_{1,0}\tau_2}{\tau_2}&=\lim _{t_{-1,0}\rightarrow 0}
\frac{\sum_{i,j=1,i< j}^3 (C_{i}S_{j}\lambda_{j}^{\frac12}-C_{j}S_{i}\lambda_{i}^{\frac12})
(\lambda_j^2-\lambda_i^2)E_{i}E_{j}}{\sum_{i,j=1,i< j}^3 (C_{i}S_{j}\lambda_{j}^{\frac12}-C_{j}S_{i}\lambda_{i}^{\frac12})(\lambda_j-\lambda_i)E_{i}E_{j}}\\
&=\lim _{t_{-1,0}\rightarrow 0}\frac{\sum_{i,j=1,i< j}^3 C_{i}C_{j}
(\lambda_j^2-\lambda_i^2)E_{i}E_{j}}{\sum_{i,j=1,i< j}^3 C_{i}C_{j}(\lambda_j-\lambda_i)E_{i}E_{j}}\\
&=\frac{\sum_{i,j=1,i< j}^3
(\lambda_j^2-\lambda_i^2)E_{i}E_{j}}{\sum_{i,j=1,i< j}^3 (\lambda_j-\lambda_i)E_{i}E_{j}},
\end{align*}
therefore $t_{-1,0}=0$ is a removable singular point.
Using the formula \eqref{aij} for $a_{i,j}$ of the Lax matrix, we have
\begin{align*}
a_{1,1}&=\d_{1,0}\ln\tau_1
=\frac{\lambda_1C_{1}E_{1}+\lambda_2C_{2}E_{2}+\lambda_3C_{3}E_{3}}{C_{1}E_{1}+C_{2}E_{2}+C_{3}E_{3}},\\
a_{2,2}&=\d_{1,0}\ln\frac{\tau_2}{\tau_1}=\frac{\sum_{i,j=1,i< j}^3 (C_{i}S_{j}\lambda_{j}^{\frac12}-C_{j}S_{i}\lambda_{i}^{\frac12})(\lambda_j^2-\lambda_i^2)E_{i}E_{j}}{\sum_{i,j=1,i< j}^3 (C_{i}S_{j}\lambda_{j}^{\frac12}-C_{j}S_{i}\lambda_{i}^{\frac12})(\lambda_j-\lambda_i)E_{i}E_{j}}-a_{11},\\
a_{3,3}&=\d_{1,0}\ln\frac{\tau_3}{\tau_2}
=\sum_{i=1}^3\lambda_i-(a_{1,1}+a_{2,2}).
\end{align*}
Assuming the ordering for the eigenvalues as
\begin{equation*}
\lambda_1<\lambda_2<\lambda_3,
\end{equation*}
one can obtain the following asymptotic {\bf sorting property}  of the Lax matrix,
\begin{align*}\L&\longrightarrow
\left(\begin{matrix}
\lambda_3&1&0\\
0&\lambda_2&1\\
0&0&\lambda_1\\
\end{matrix}\right),\qquad{\rm for}\quad t_{-1,0}\to \infty,\\[2.0ex]
\L&\longrightarrow\begin{cases}
\left(\begin{matrix}
\lambda_3&1&0\\
0&\lambda_2&1\\
0&0&\lambda_1\\
\end{matrix}\right),\qquad{\rm for}\quad t_{1,0}\to\infty,\\
\left(\begin{matrix}
\lambda_1&1&0\\
0&\lambda_2&1\\
0&0&\lambda_3\\
\end{matrix}\right),\qquad{\rm for}\quad t_{1,0}\to-\infty.
\end{cases}
\end{align*}

To see the orbit generated by the solution, we consider the projection $\pi$ of the Lax matrix
on the diagonal part, i.e.
\begin{equation*}
\pi : \quad\L=\begin{pmatrix}
a_{1,1}& 1 & 0\\
a_{2,1}& a_{2,2} & 1\\
a_{3,1}& a_{3,2} & a_{3,3}
\end{pmatrix}\quad \longmapsto \quad {\rm diag}(\L)\equiv (a_{1,1},a_{2,2},a_{3,3}).
\end{equation*}
Figure \ref{x=1bth} illustrates the image of the map $\pi$:  The left panel (a) shows the image
in the case of $(1,2)$-BTH  for $-5\le t_{0,0}\le 5$ and $-5\le t_{0,1}\le 5$.
The right panel (b) of the figure shows the case of the original Toda lattice, that is, the corresponding
Lax matrix is $3\times 3$-sized tridiagonal matrix. That example gives following proposition.
\begin{proposition} In the case of the original Toda lattice, one can show that
all the orbits have to cross the center point $(\lambda_2,\sigma_1-2\lambda_2,\lambda_2)$ with $\sigma_1=\sum_{j=1}^3\lambda_j$.
 However, the orbits for the $(1,2)$-BTH have no such restriction, but those
still go through the points close to the center.
\end{proposition}
\begin{proof}
Firstly because the Lax matrix has eigenvalues $\lambda_1, \lambda_2, \lambda_3$ with $\lambda_1<\lambda_2<\lambda_3$. We set the diagonal elements $(a_{1,1},a_{2,2},a_{3,3})$ take values $(\Delta_1,\sigma_1-2\Delta_1,\Delta_1)$.
Then considering matrix
\begin{equation*}
\begin{pmatrix}
\Delta_1& 1 & 0\\
a_{2,1}& \sigma_1-2\Delta_1 & 1\\
a_{3,1}& a_{3,2} &\Delta_1
\end{pmatrix}
\end{equation*}
has eigenvalues $\lambda_1, \lambda_2, \lambda_3$ with $\lambda_1<\lambda_2<\lambda_3$,
we get
\begin{eqnarray*}\label{general identity0}
(\Delta_1-\la)[(\sigma_1-2\Delta_1-\la)(\Delta_1-\la)-a_{3,2}]-a_{2,1}(\Delta_1-\la)+a_{3,1}=-\la^3+\sigma_1\la^2-\sigma_2\la+\sigma_3,
\end{eqnarray*}
where $\sigma_1,\sigma_2,\sigma_3$ are first three fundamental
symmetric polynomials.
 Equation above implies
\begin{eqnarray}\label{general identity1}
a_{3,2}+a_{2,1}=2\sigma_1\Delta_1-3\Delta_1^2-\sigma_2,
\end{eqnarray}
\begin{eqnarray}\label{general identity2}
\Delta_1^2\sigma_1-2\Delta_1^3-(a_{3,2}+a_{2,1})\Delta_1+a_{3,1}=\sigma_3,
\end{eqnarray}
which further leads to
\begin{eqnarray}\label{crossing equation}
\Delta_1^3-\sigma_1\Delta_1^2+\sigma_2\Delta_1-\sigma_3+a_{3,1}=0.
\end{eqnarray}
By eq.\eqref{aij}, we get
\begin{equation}\label{aij1}
a_{3,2}(t)=\frac{P_{2}(\hat D_L)\tau_2\circ\tau_{2}}{\tau_{2}\tau_2}
=\frac{P_{1}(\hat D_R)\tau_3\circ\tau_{1}}{\tau_{2}\tau_2}.
\end{equation}
Because we choose all eigenvalues to be positive,
\begin{equation}\label{aij2}
D_{1,0}^2\tau_2\circ\tau_{2}=\sum_{1\leq i,j,k\leq
3}(C_{i}S_{j}\lambda_{j}^{\frac12}-C_{j}S_{i}\lambda_{i}^{\frac12})(C_{j}S_{k}\lambda_{k}^{\frac12}-C_{k}S_{j}\lambda_{j}^{\frac12})
(\la_k-\la_j)(\la_j-\la_i)(\la_k+\la_i+2\la_j)>0,
\end{equation}
therefore $a_{3,2}(t)>0.$
Similarly by eq.\eqref{aij}, we get
\begin{equation}\label{aij1'}
a_{2,1}(t)=\frac{P_{2}(\hat D_L)\tau_1\circ\tau_{1}}{\tau_{1}\tau_1}
=\frac{P_{1}(\hat
D_R)\tau_2\circ\tau_{0}}{\tau_{1}\tau_1}=\frac{\d_{-1,0}\tau_2}{\tau_{1}\tau_1}.
\end{equation}
Because of eq.\eqref{derivativepositive}, we get
\begin{equation}\label{aij2'}
a_{2,1}(t)>0.
\end{equation}
By eq.\eqref{general identity1}, we can get the following identity
must be correct
\begin{eqnarray}\label{general identity1'}
f(\Delta_1):=2\sigma_1\Delta_1-3\Delta_1^2-\sigma_2>0.
\end{eqnarray}
 For triangonal Toda hierarchy, $a_{3,1}=0$, and
eq.\eqref{crossing equation} has there roots $\lambda_1, \lambda_2,
\lambda_3$ with order $\lambda_1<\lambda_2<\lambda_3$. Let
$\Delta_1=\la_3$, we find
\begin{eqnarray}\label{general identity1''}
f(\la_3)&=&2\sigma_1\la_3-3\la_3^2-\sigma_2=2(\la_1+\la_2+\la_3)\la_3-3\la_3^2-(\la_1\la_2+\la_1\la_3+\la_2\la_3)\\\nonumber
&=&(\la_1-\la_3)(\la_3-\la_2)<0,
\end{eqnarray}
which is in contradict with eq.\eqref{general identity1'}. Similarly
we can also find $\Delta_1=\la_1$ is also in contradict with
eq.\eqref{general identity1'}. Therefore the only choice is
$\Delta_1=\la_2$.
 That is why  for  the original Toda lattice,
all the orbits have to cross the center point
$(\lambda_2,\sigma_1-2\lambda_2,\lambda_2)$ with
$\sigma_1=\sum_{j=1}^3\lambda_j$ as shown in Fig.\ref{x=1bth}(b).
 By eq.\eqref{aij}, we get
\begin{equation}\label{aij1'}
a_{3,1}(t)=\frac{\tau_3\tau_{0}}{\tau_{2}\tau_1}.
\end{equation}
So for the $(1,2)$-BTH, $a_{3,1}(t)$ is always positive when
$t_{-1,0}>0$ because of the positivity of tau functions. From that
and considering eq.\eqref{crossing equation}, we can get that there
will be one more part which is close to the crossing point
$(\lambda_2,\sigma_1-2\lambda_2,\lambda_2)$ just as
Fig.\ref{x=1bth}(a).
\end{proof}

\begin{figure}[h!]
\centering
\raisebox{0.85in}{(a)}\includegraphics[scale=0.48]{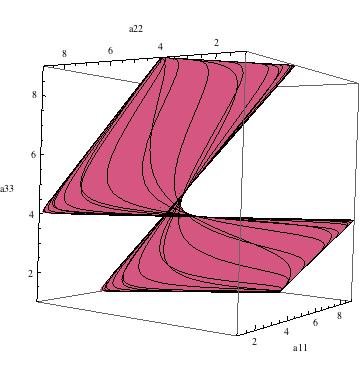}
\hskip 0.3cm
\raisebox{0.85in}{(b)}\raisebox{-0.1cm}{\includegraphics[scale=0.48]{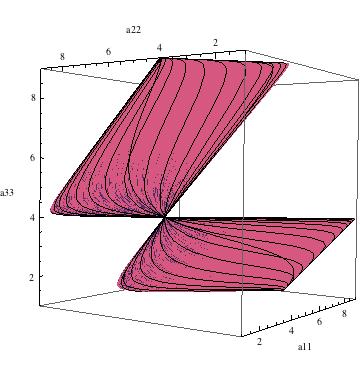}}
 \caption{The image $\pi(\L)=(a_{1,1},a_{2,2},a_{3,3})$:  (a) for the $(1,2)$-BTH, and (b) for the tridiagonal Toda hierarchy.  Those orbits are obtained by changing $t_{0,0}$ and $t_{0,1}$ with fixed $t_{-1,0}=1$. The eigenvalues are given by $\lambda_1=1, \lambda_2=4, \lambda_3=9$. }\label{x=1bth}
\end{figure}
The boundaries of Fig.\ref{x=1bth}(a) are characterized by Fig.\ref{x=20,z=0.eps}(a) and Fig.\ref{x=20,z=0.eps}(b).
\begin{figure}[h!]
\centering
\raisebox{0.85in}{(a)}\includegraphics[scale=0.48]{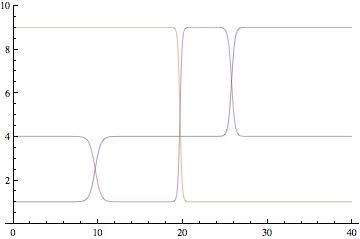}
\hskip 0.3cm
\raisebox{0.85in}{(b)}{\includegraphics[scale=0.48]{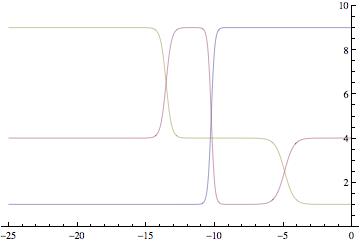}}
 \caption{ Graphs for $(a_{1,1},a_{2,2},a_{3,3})$ of the (1,2)-BTH:  (a)  depending on the parameters $t_{0,0}$ with $t_{-1,0}=20$
 and $t_{0,1}=-2$, (b) depending on the parameter $t_{0,0}$ with  $t_{-1,0}=20 $ and $t_{0,1}=1$.  The eigenvalues are $\lambda_1=1, \lambda_2=4, \lambda_3=9$ and $A_1= A_2= A_3=1$.}\label{x=20,z=0.eps}
\end{figure}
Fixing another time variable $t_{1,1}$, their two dimensional graphs are as Fig.\ref{y=-0.4,z=0}(a) and Fig.\ref{y=-0.4,z=0}(b).
\begin{figure}[h!]
\centering
\raisebox{0.85in}{(a)}
{\includegraphics[scale=0.48]{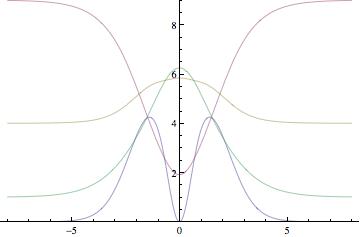}}
\hskip 0.3cm
\raisebox{0.85in}{(b)}{\includegraphics[scale=0.48]{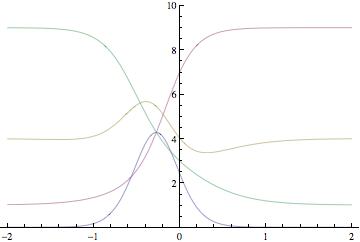}}
 \caption{ Graphs for $(a_{3,1},a_{1,1},a_{2,2},a_{3,3})$ for (1,2)-BTH  (a)depending on parameters $t_{-1,0}$ with $t_{0,1}=0$ and  $ t_{0,0}=-0.4$ , (b) depending on parameters $t_{0,0}$ with $t_{0,1}=0$ and  $ t_{-1,0}=1$, Here $\lambda_1=1, \lambda_2=4, \lambda_3=9, A_1= A_2= A_3=1$.}\label{y=-0.4,z=0}
\end{figure}

After these geometric pictures, moment map \cite{ISo} related to the $(1,2)$-BTH will be given in the next subsection.

\subsection{Moment polytope for the $(1,2)$-BTH}
From the last section, $(\tau_1,\tau_2)$ have following form
\begin{equation*}
\tau_1=\sum_{k=1}^3 C_{k}E_{k},
\end{equation*}

\begin{align*}
\tau_2
&=\sum_{i,j=1,i< j}^3 (C_{i}S_{j}\lambda_{j}^{\frac12}-C_{j}S_{i}\lambda_{i}^{\frac12})(\lambda_j-\lambda_i)E_{i}E_{j}.
\end{align*}

We can treat $(\tau_1,\tau_2)$ as one point of Flag manifold $G/B$,
where $G:=Sl(3,R)$ and
$B$ is a Borel subgroup(upper triangular subgroup) containing the Cartan Lie subgroup of $G$(diagonal torus).
Here we describe the moment polytope
defined by the map \cite{topology},
\begin{align*}
\mu:\quad &G/B\longrightarrow \H^*\\
&(\tau_1,\tau_2)\longmapsto M_{\tau_1}+M_{\tau_2},
\end{align*}
where
   $\H^*$ is the dual of the Cartan Lie subgroup of $G$,
\begin{align*}
 M_{\tau_1}&=\frac{C_{1}E_{1}L_1+C_{2}E_{2}L_2+C_{3}E_{3}L_3}{C_{1}E_{1}+C_{2}E_{2}+C_{3}E_{3}},\\
 M_{\tau_2}&=\frac{\sum_{i,j=1,i< j}^3 (C_{i}S_{j}\lambda_{j}^{\frac12}-C_{j}S_{i}\lambda_{i}^{\frac12})(\lambda_j-\lambda_i)E_{i}E_{j}(L_i+L_j)}{\sum_{i,j=1,i< j}^3 (C_{i}S_{j}\lambda_{j}^{\frac12}-C_{j}S_{i}\lambda_{i}^{\frac12})(\lambda_j-\lambda_i)E_{i}E_{j}}.
\end{align*}
Here the weight vectors $L_i$'s are defined by
\begin{align*}
L_1:=(1,0), \quad L_2:=\frac12\,(-1, \sqrt{3}),\quad  L_3:=\frac12\,(-1, -\sqrt{3}).
\end{align*}
Fig. \ref{x=1hexagon} illustrates the moment polytope (i.e. the graph of $ M_{\tau_1}+M_{\tau_2}$) for our example.
\begin{figure}[h!]
\centering
\raisebox{0.85in}{}\includegraphics[scale=0.43]{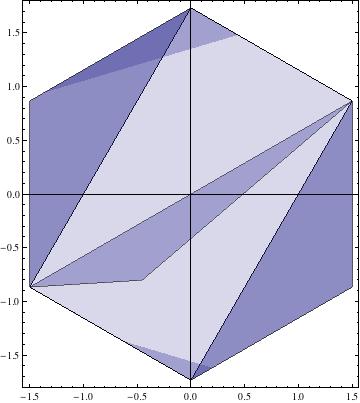}
 \caption{The moment polytope (i.e. $ M_{\tau_1}+M_{\tau_2}$ in $\H^*$) of the (1,2)-BTH: The orbits are obtained by changing $t_{0,0}$ and $t_{0,1}$ with fixed $t_{-1,0}=1$.  The eigenvalues are $\lambda_1=1, \lambda_2=4, \lambda_3=9$.}\label{x=1hexagon}
\end{figure}
In the Fig.\ref{x=1hexagon}, the vertex $\frac12(3,\sqrt{3})$ represent the highest weight $L_1-L_3$ which is the starting point of this $(1,2)$-BTH flow($t_{0,1}$ flow,$t_{0,0}$ flow). The vertex $\frac12(-3,-\sqrt{3})$ represent the lowest weight $-L_1+L_3$ which is the destination of the flow. The boundaries of Fig. \ref{x=1hexagon} correspond to the $\bar G=Sl(2,R)$ for $(1,2)$-BTH associated with two of  $a_{2,1},a_{3,1}$ and $a_{3,2}$ are zeroes. The six vertices in Fig. \ref{x=1hexagon} are fixed points of Cartan subgroup of $G$ which can be generated from the highest weight by the action of the symmetric group $S_3$. These six vertices are in bijection with the elements of the weyl group.

  Till now, we find the moment polytope can tell how the flow of time variables of BTH goes well, meanwhile  the sorting property of the flows and representation in  Lie algebra describing the orbit are also seen from the  Fig. \ref{x=1hexagon}. What about the polytope corresponding to higher rank matrix-formed BTH is an interesting question.

\sectionnew{ $(N,1)$-Bigraded Toda hierarchy  }
In last section, $(1,M)$-BTH are introduced a lot. In this section, we concentrate on the $(N,1)$-BTH. For convenience of calculation, we firstly use interpolated form of the BTH.

In the following part, we will introduce some concrete primary flows of the $(N,1)$-BTH.

{\bf $(2,1)$-BTH:} The Lax operator of $(2,1)$-BTH is as following
\begin{eqnarray}
\L_{2,1}= \La^{2}+ u_1\La+ u_0+ u_{-1}\La^{-1}.
\end{eqnarray}
The equations \eqref{BTHLax}in this case are as follows
\begin{eqnarray}
&& \partial_{2,0} \L_{2,1}= [\Lambda +(1+\La)^{-1} u_1(x), \L_{2,1}],
\end{eqnarray}
which further lead to the following concrete equations
\begin{eqnarray}\label{N=2,M=12,0flow}
\begin{cases}
 \partial_{2,0}  u_1(x)=  u_1(x+\epsilon)- u_1(x)+ u_1(x)(1-\La)(1+\La)^{-1} u_1(x),\\
\partial_{2,0}  u_0(x)=  u_{-1}(x+\epsilon)- u_{-1}(x),\\
\partial_{2,0}  u_{-1}(x)=  u_{-1}(x)(1-\La^{-1})(1+\La)^{-1} u_1(x),
 \end{cases}
\end{eqnarray}
which is equivalent to eq.$(40)$ in \cite{svinin} and also related to the system (10)-(12) proposed in \cite{YTWuXBHuJPA}.

{\bf $(3,1)$-BTH:} The equations of $(3,1)$-BTH is as following
\begin{eqnarray}\label{N=3,M=13,0flow}
\begin{cases}
\partial_{3,0}  u_2(x)=  u_1(x+\epsilon)- u_1(x)+ u_2(x)(1-\La^2)(1+\La+\La^2)^{-1} u_2(x),\\
 \partial_{3,0}  u_1(x)=  u_0(x+\epsilon)- u_0(x)+ u_1(x)(1-\La)(1+\La+\La^2)^{-1} u_2(x),\\
\partial_{3,0}  u_0(x)=  u_{-1}(x+\epsilon)- u_{-1}(x),\\
\partial_{3,0}  u_{-1}(x)=  u_{-1}(x)(1-\La^{-1})(1+\La+\La^2)^{-1} u_2(x),
 \end{cases}
\end{eqnarray}
which can be rewritten as
\begin{eqnarray}\label{N=3,M=13,0flow'}
\begin{cases}
\partial_{3,0}(\sum_{i=0}^2  v_2(x+i\ep))=  v_1(x+\epsilon)- v_1(x)+(\sum_{i=0}^2  v_2(x+i\ep))(1-\La^2) v_2(x),\\
 \partial_{3,0}   v_1(x)=  v_0(x+\epsilon)- v_0(x)+  v_1(x)(1-\La) v_2(x),\\
\partial_{3,0}  v_0(x)=  v_{-1}(x+\epsilon)- v_{-1}(x),\\
\partial_{3,0} v_{-1}(x)=  v_{-1}(x)(1-\La^{-1}) v_2(x),
 \end{cases}
\end{eqnarray}
by transformation $u_2=(1+\La+\La^2) v_2;\ u_i=v_i, i=1,0,-1.$
This is the case when  $n=1,m=-2,l=4$ for eq.(10) in \cite{svinin} after  treating $v(x+i\ep)$ as $v(i)$,
 \begin{equation}
\begin{array}{l}
\displaystyle
\sum_{s=1}^{3}v_{2}^{\prime}(i+s-1)=
\sum_{s=1}^{3}v_{2}(i+s-1)
\times\left(v_{2}(i)-
v_{2}(i+2)\right)
+ v_{1}(i+1) - v_1(i), \\[0.4cm]
\displaystyle
v_{k}^{\prime}(i)=
v_{k}(i)\left(\sum_{s=1}^{2}v_{0}(i+s-1)-
\sum_{s=1}^{2}v_{0}(i+s-1+k)\right) \\[0.4cm]
\displaystyle
+ v_{k-1}(i-1) - v_{k-1}(i),\;\;\;
k = 1,0,-1.
\end{array}
\end{equation}

To generalize the results above, $(N,1)$-BTH will be considered as following.

{\bf $(N,1)$-BTH:} The concrete equations of $(N,1)$-BTH is as following
\begin{eqnarray*}\label{N=3,M=13,0flow'}
\begin{cases}
\partial_{N,0}(\sum_{i=0}^{N-1}  v_{N-1}(x+i\ep))=  v_{N-2}(x+\epsilon)- v_{N-2}(x)+(\sum_{i=0}^{N-1}  v_{N-1}(x+i\ep))(1-\La^{N-1}) v_{N-1}(x),\\
\partial_{N,0} v_{N-2}(x)=  v_{N-3}(x+\epsilon)- v_{N-3}(x)+ v_{N-2}(x)(1-\La^{N-2}) v_{N-1}(x),\\
\dots \dots \dots \dots \dots \dots \\
 \partial_{N,0}   v_1(x)=  v_0(x+\epsilon)- v_0(x)+  v_1(x)(1-\La) v_2(x),\\
\partial_{N,0}  v_0(x)=  v_{-1}(x+\epsilon)- v_{-1}(x),\\
\partial_{N,0} v_{-1}(x)=  v_{-1}(x)(1-\La^{-1}) v_2(x),
 \end{cases}
\end{eqnarray*}
where $u_{N-1}=(1+\La+\dots+\La^{N-1}) v_{N-1};\ u_i=v_i, i=N-2,\dots,0,-1.$

Similarly we can rewrite the equations of $(N,1)$-BTH as following
 \begin{equation}\label{N1}
\begin{cases}
\sum_{s=1}^{N}\partial_{N,0}v_{N-1}(i+s-1)=
\sum_{s=1}^{N}v_{N-1}(i+s-1)
\times\left(v_{N-1}(i)-
v_{N-1}(i+N-1)\right)  \\
+ v_{N-2}(i+1) - v_{N-2}(i), \\
\partial_{N,0}v_{k}(i)=
v_{k}(i)\left(v_{N-1}(i)-
v_{N-1}(i+k)\right)
+ v_{k-1}(i+1) - v_{k-1}(i),
\;\;
k = N-2,N-3,..., -1,
\end{cases}
\end{equation}
where $v_j(i):=v_j(x+i\ep)$ and $v_{-2}=0$.

This is exactly the case when $n=1,m=-N+1,l=N+1$ for eq.(10) in \cite{svinin}. It has another Lax representation which will be discussed in the next section.

\sectionnew{Another Lax representation for primary flows of $(N,1)$-BTH }
In this section, we will introduce another Lax representation of $\d_{N,0}$ flow(i.e. the primary flow) of  $(N,1)$-BTH.

For an arbitrary pair of integers $n\in{\bf N}$  and $m \le n-1,$ we
define infinite collection of first-order differential operators
\begin{equation}
H_i = \partial_{N,0} - v_{N-1}(i, t_{N,0}),\;\;
i\in\Z,
\end{equation}
and \[
G_i = \partial_{N,0} + \sum_{k=1}^{N-1}v_{N-1}(i-k, t_{N,0}) +
\sum_{k=1}^{N}v_{N-1-k}(i-N+1, t_{N,0})H_{i-k}^{-1}...H_{i-2}^{-1}H_{i-1}^{-1}.
\]

Let us define following auxiliary equations on infinite collection
of dressing operators $\{\hat{w}_i,\; i\in\Z\}$:
\begin{equation}
G_i\hat{w}_i = \hat{w}_{i+1-N}\partial_{N,0},\;\;\;
H_i\hat{w}_i = \hat{w}_{i+1}\partial_{N,0},
\label{aux1}
\end{equation}
which can be rewritten in terms of Baker-Akhiezer(BA) function $\psi_i$ as
\begin{equation}
G_i\psi_i = z\psi_{i+1-N},\;\;\;
H_i\psi_i = z\psi_{i+1}.
\label{aux2}
\end{equation}
We can show that the compatibility conditions of eq.\eqref{aux2}
are well-determined system of equations for the fields $\{v_{-1}(i, t_{N,0}),  v_0(i, t_{N,0}),..., v_{N-1}(i, t_{N,0})\}$.

Formally,
consistency condition of (\ref{aux2}) is given by

\[G_{i+1}H_i=H_{i-N+1}G_i.\]

The technical observation will lead to  that eq.(\ref{aux2}) can be rewritten in terms of
$(L,  A_{N,0})$-pair
\[
L(\psi_i) = z\psi_i,\;\;\;  \partial_{N,0}\psi_i = A_{N,0}(\psi_i)
\]
where $L$ and $ A_{N,0}$ are
difference operators acting on BA functions
$\{\psi_i,\;\; i\in\Z\}$ as
\begin{equation}
\begin{array}{l}
\displaystyle
L(\psi_i) = z\psi_{i+N} + \left(\sum_{s=1}^{N}v_{N-1}(i+s-1)\right)\psi_{i+N-1} +
\sum_{j=1}^{N}\frac{1}{z^j}v_{N-1-j}(i)\psi_{i+N-1-j}, \\[0.4cm]
\displaystyle
A_{N,0}(\psi_i) = z\psi_{i+1} + v_{N-1}(i)\psi_{i}.
\end{array}
\label{la0}
\end{equation}
That means
\begin{equation}
\begin{array}{l}
\displaystyle
L = G_{i+N}^{-1}H_{1}\partial_{N,0}+ \left(\sum_{s=1}^{N}v_{N-1}(i+s-1)\right) G_{i+N-1}^{-1} +
\sum_{j=1}^{N}\frac{1}zv_{N-1-j}(i) G_{i+N-1-j}^{-1}H_{i-j}^{-1}...H_{i-2}^{-1}H_{i-1}^{-1}, \\[0.4cm]
\displaystyle
A_{N,0} = H_{1} +v_{N-1}(i).
\end{array}
\end{equation}
Then consistency conditions of eq.(\ref{aux2}) are expressed in a form
of the Lax equation
\begin{equation}
\partial_{N,0}L = [A_{N,0},  L] =A_{N,0}L-  L A_{N,0},
\end{equation}
which exactly leads to eq.\eqref{N1}.

 Till now we have given another Lax construction for primary flows of the $(N,1)$-BTH.

\sectionnew{Lattice Miura transformation}
As we all know, many one-field lattice equations are very useful in a lot of branches of science such as biology, medical science, physics and so on.
In this section, we will introduce one kind of Miura mapping which connects one-field lattice equation with $(N+1)$-field ones(i.e. $(N,1)$-BTH).

Define
\[
F_i=G_{i+N}H_{i+N-1}\dots H_{i+1}H_i,\ \ i\in \Z\]
which is a N-order differential operator, we obtain

\[\label{system}
F_i \psi_i=z^{N+1}\psi_{i+1},\ \ \ H_i\psi_i=z\psi_{i+1}, \ \ i\in \Z.\]

Because
\begin{equation}
G_i\psi_i = z\psi_{i+1-N},\;\;\;
H_i\psi_i = z\psi_{i+1},\ \ i\in \Z,
\end{equation}
we define
\begin{equation}\label{barsystem}
\bar G_i\bar\psi_i = z\bar\psi_{i+1},\;\;\;
\bar H_i\bar\psi_i = z\bar\psi_{i+N+1},\;\;\; \psi_i = \bar\psi_{(N+1)i},
\end{equation}
where
\[
\bar H_i=\d_{N,0}-\sum_{k=1}^{N+1} r_{i+k-1},\ \ \
\bar G_i=\d_{N,0}- r_{i}.\]

The compatibility of system \eqref{barsystem} leads to the following one-field lattice equations
\begin{equation}
\begin{array}{l}
\displaystyle
\sum_{s=1}^{N}\partial_{N,0}r_{i+s-1}
=
\sum_{s=1}^{N}r_{i+s-1}
\times\left(r_{i+N}-
r_{i-1}\right),\ \ i\in \Z.
\end{array}
\end{equation}

Define
\[
\bar F_i=\bar G_{i+N}\bar G_{i+N-1}\dots\bar G_{i+1}\bar G_{i},\]
which is a N-order differential operator. Now  we  consider a new system

\[\label{barsystem2}
\bar F_i \bar\psi_i=z^{N+1}\bar\psi_{i+N+1},\ \ \ \bar H_i\bar\psi_i=z\bar\psi_{i+N+1}.\]
Comparing system \eqref{barsystem} with \eqref{system} will lead to following identification
\begin{equation}\label{Fi}
F_i = \bar F_{(N+1)i},\;\;\;
H_i =\bar H_{(N+1)i},
\end{equation}
which will tell us the Miura transformation in detail.

Eq.\eqref{Fi} can be rewritten in the following equivalent form
\[
G_{i+N}H_{i+N-1}\dots H_{i+1}H_i=\bar G_{(N+1)i+N}\bar G_{(N+1)i+N-1}\dots\bar G_{(N+1)i+1}\bar G_{(N+1)i}.\]

In the following, we will give two specific examples for original Toda hierarchy and $(2,1)$-BTH including their corresponding lattice Miura transformation and one field equations.

\begin{example}

For $(1,1)$-BTH, i.e. original Toda hierarchy, the case $N=1,l=2,n=1,m=0,\bar n=2,\bar m=1$ in \cite{svinin}
\begin{eqnarray}\label{N=1,M=1,0flow}
\begin{cases}
\partial_{1,0}  v_0(x)=  v_{-1}(x+\epsilon)- v_{-1}(x),\\
\partial_{1,0}  v_{-1}(x)=  v_{-1}(x)(1-\La^{-1}) v_1(x),
 \end{cases}
\end{eqnarray}
which can be rewritten as following discrete form
\begin{eqnarray}\label{N=1,M=1,0flow0}
\begin{cases}
\partial_{1,0}  v_0(i)=  v_{-1}(i+1)- v_{-1}(i),\\
\partial_{1,0}  v_{-1}(i)=  v_{-1}(i)( v_1(i)-v_1(i-1)),
 \end{cases}
\end{eqnarray}
by transformation $v_{j}(i):=v_{j}(x+i\ep),j=0,-1.$
The lattice Miura transformation is as following
\begin{equation}\label{lattice}
\begin{cases}
v_{1}(i)
&=r_{2i}+ r_{2i+1},\\
v_0(i)
&=r_{2i-1}r_{2i}.
\end{cases}
\end{equation}

After lattice Miura transformation \eqref{lattice}, the $t_{1,0}$ flow of $(1,1)$-BTH can be transformed into the following
one-field equation, i.e. the Volterra lattice which is a very useful ecological competition model  in biology,
\begin{equation}\label{lattice2}
\partial_{1,0}r_{i}=r_{i}(r_{i+1}-
r_{i-1}),\ \ i\in \Z.
\end{equation}
\end{example}

\begin{example}

For $(2,1)$-BTH,  $N=2,l=3,n=1,m=-1,\bar n=3,\bar m=1,$

the lattice Miura transformation is as following
\begin{equation}\label{latticeMiura}
\begin{cases}
v_{1}(i)
&=r_{3i+2}+ r_{3i+1}+ r_{3i},\\
v_0(i)
&=r_{3i-1}r_{3i}+r_{3i-1}r_{3i+1}+r_{3i}r_{3i+2}+r_{3i+1}r_{3i+3}+r_{3i+2}r_{3i+3}\\&
+ r_{3i+2}^2+r_{3i+1}^2+ r_{3i}^2+2r_{3i+2}r_{3i+1}+ 2r_{3i}r_{3i+1},\\
v_{-1}(i)
&=
-r_{3i-3}r_{3i-2}
r_{3i-1} +r_{3i-2}r_{3i-1}r_{3i-1}\\
&+(r_{3i-2}+r_{3i-1})(r_{3i-1}+r_{3i})(r_{3i+1}+r_{3i}),
\end{cases}
\end{equation}

After lattice Miura transformation \eqref{latticeMiura}, the $t_{2,0}$ flow of $(2,1)$-BTH can be transformed into the following
one-field equation
\begin{equation}\label{lattice3}
\partial_{2,0}(r_{i}+r_{i+1})=(r_{i}+r_{i+1})(r_{i+2}-
r_{i-1}),\ \ i\in \Z.
\end{equation}

\end{example}
The relation between $(N,1)$-BTH and one-field lattice equations give us some hints on how to get solutions of one-field lattice equations from the solutions of BTH (known in \cite{solutionBTH}) by certain transformation.

\sectionnew{Conclusions and discussions}
 We give finite dimensional exponential  solutions of the bigraded Toda Hierarchy(BTH). As
 a specific example of exponential solutions of the BTH, we consider a regular solution for the $(1,2)$-BTH with $3\times 3$ Lax matrix. The difference between $(1,2)$-BTH and original Toda hierarchy is found from a  geometric viewpoint by diagonal projection and moment map.
 Our future work contains finding other regular solutions corresponding to other cases of BTH and their geometric description.
 After that, we construct another Lax representation of
bigraded Toda hierarchy(BTH) and introduce lattice Miura transformation of BTH.
These
Miura transformations give a good connection between primary equation of $(N,1)$-BTH  and one-field lattice equations which include Volterra lattice equation.
What kinds of one-field lattice equations will correspond to the whole hierarchies in BTH is  an interesting questions.

{\bf {Acknowledgments:}}
  {\small
  This work was partly carried out under the suggestion  of Professor Yuji Kodama. Chuanzhong Li would like to thank Professor Yuji Kodama for his suggestion and many useful discussions.  }



\end{document}